\documentclass[a4paper,12pt]{article}
\usepackage{amsmath,amssymb,amsfonts,amsthm}
\newtheorem{theorem}{Theorem}[section]
\newtheorem{lemma}[theorem]{Lemma}
\newtheorem{definition}[theorem]{Definition}

\newtheorem{proposition}[theorem]{Proposition}

\usepackage{mathabx}

\usepackage{graphicx}

\newcommand{\be}{\begin{equation}}
\newcommand{\ee}{\end{equation}}
\newcommand{\bea}{\begin{eqnarray}}
\newcommand{\eea}{\end{eqnarray}}

\newcommand{\nn}{\nonumber}

\newcommand{\ben}{\begin{displaymath}}

\newcommand{\een}{\end{displaymath}}

\newcommand{\Su}{\mathcal{S}}
\newcommand{\R}{\mathbb R}

\newcommand{\n}{\noindent}
\newcommand{\vs}{\vspace{0.2cm}}
\newcommand{\field}[1]{\mathbb{#1}}

\title{Area-charge inequality for black holes}

\author{Sergio Dain$^{1,2}$,  Jos\'e Luis Jaramillo$^{2,3}$ 
and Mart\'\i n Reiris$^{2}$\\
\\
$^1$ Facultad de Matem\'atica, Astronom\'\i a y F\'\i sica, FaMAF, \\
Universidad Nacional de C\'ordoba \\
Instituto de F\'\i sica Enrique Gaviola, IFEG, CONICET, \\
Ciudad Universitaria (5000) C\'ordoba, Argentina\\
$^2$ Max-Planck-Institut f{\"u}r Gravitationsphysik, Albert Einstein\\
Institut, Am M\"uhlenberg 1 D-14476 Potsdam Germany \\
$^3$ Laboratoire Univers et Th\'eories (LUTH), Observatoire de Paris, \\
CNRS, Universit\'e Paris Diderot, 92190 Meudon, France}

\begin{document}
\maketitle

\begin{abstract}
  The inequality between area and charge $A\geq 4\pi Q^2$ for dynamical black
  holes is proved. No symmetry assumption is made and charged matter fields are
  included. Extensions of this inequality are also proved for regions in the
  spacetime which are not necessarily black hole boundaries.
\end{abstract}

\section{Introduction}
\label{s:introduccion}
In a recent series of articles \cite{dain10d} \cite{Acena:2010ws}
\cite{Dain:2011pi} \cite{Jaramillo:2011pg} the quasi-local inequality between
area and angular momentum was proved for dynamical axially symmetric black
holes (see also \cite{Ansorg:2007fh} \cite{Hennig:2008zy}
\cite{hennig08}   \cite{Ansorg:2010ru}
for a proof in the stationary case).  In these articles the assumption of axial
symmetry is essential since it provides a canonical notion of quasi-local
angular momentum.  The natural question is whether similar kind of inequalities
hold without this symmetry assumption, that certainly restricts their
application in physically realistic scenarios. A natural first step to answer
this question is to study the related inequality involving the electric
charge, since the charge is always well defined as a quasi-local quantity.

In \cite{Gibbons:1998zr} the expected inequality for area and charge has been
proved for stable minimal surfaces on time symmetric initial data. The main
goal of this article is to extend this result in several directions.  First,
we prove the inequality for generic dynamical black holes. Second, we also
prove versions of this inequality for regions which are not necessarily black
hole boundaries, that is, regions that can be interpreted as the boundaries of
ordinary objects.

The plan of the article is the following. In section \ref{s:main} we present
our mains results which are given by theorems \ref{t:main1}, \ref{t:main2} and
\ref{t:main3}. We also discuss in this section the physical meaning of these
results. In section \ref{s:black-holes} we prove theorem \ref{t:main1} and in
section \ref{s:regions} we prove theorems \ref{t:main2} and \ref{t:main3}.

\section{Main result}
\label{s:main}

Consider Einstein equations with cosmological constant $\Lambda$
\begin{equation}
  \label{eq:3a}
  G_{ab}=8\pi (T^{EM}_{ab}+T_{ab})-\Lambda g_{ab},
\end{equation}
where $T^{EM}_{ab}$ is the electromagnetic energy-momentum tensor given by
\begin{equation}
  \label{eq:4}
  T^{EM}_{ab}= \frac{1}{4\pi}\left(F_{ac}
    F_b{}^{c}-\frac{1}{4}g_{ab} F_{cd} F^{cd}  \right),
\end{equation}
and $F_{ab}$ is the (antisymmetric) electromagnetic field tensor. The electric and magnetic
charge of an arbitrary closed, oriented, two-surface $\Su$ embedded in the
spacetime are defined by
\begin{equation}
  \label{eq:4.5}
Q_E=\frac{1}{4\pi}\int_{\Su} {}^*\!F_{ab}, \quad 
Q_M = \frac{1}{4\pi}\int_{\Su} F_{ab},
\end{equation}
where ${}^*\!F_{ab}=\frac{1}{2} \epsilon_{abcd}F^{cd}$ is the dual of $F_{ab}$
and $\epsilon_{abcd}$ is the volume element of the metric $g_{ab}$.  It is
important to emphasize that we do not assume that the matter is uncharged,
namely we allow $\nabla_a F^{ab}=-4\pi j^b\neq 0$ (which is equivalent to
$\nabla^aT^{EM}_{ab}\neq 0$). The only condition that we impose is that the
non-electromagnetic matter field stress-energy tensor $T_{ab}$ 
satisfies the dominant energy condition.

The first main result of this article is the following theorem. 
\begin{theorem}
\label{t:main1}
Given an orientable closed marginally trapped surface $\Su$ satisfying
the spacetime stably outermost condition, in a spacetime which satisfies 
Einstein
equations (\ref{eq:3a}), with non-negative cosmological constant $\Lambda$ and
such that the non-electromagnetic matter fields $T_{ab}$ satisfy the dominant
energy condition, then it holds the inequality
\begin{equation}
\label{e:inequality}
A \geq 4\pi \left(Q_E^2+Q_M^2\right) ,
\end{equation}
where $A$,  $Q_E$ and $Q_M$ are the area, electric and magnetic charges of
$\Su$ given by (\ref{eq:4.5}). 
\end{theorem}
For the definition of marginally 
trapped surfaces and the stably outermost condition see
definition \ref{d:stability} in section \ref{s:black-holes}.  This theorem
represents a generalization of the result presented in \cite{Gibbons:1998zr}
valid for stable minimal surfaces. Theorem \ref{t:main1} is the analog of the
theorem proved in \cite{Jaramillo:2011pg} for the angular momentum. The
important difference is that in theorem \ref{t:main1} no symmetry assumption is
made.  Also the proof of this result is much simpler than the one in
\cite{Jaramillo:2011pg}, we explain this in detail in section
\ref{s:black-holes}.

Although the theorem proved in \cite{Gibbons:1998zr} (which we include as
theorem \ref{PQ2} in this article) for stable minimal surfaces embedded on
maximal initial data is more restrictive than theorem \ref{t:main1}, it is
geometrically interesting and it has also relevant applications as the ones
presented below.  One important consequence of theorem \ref{PQ2} is that it
allows a suitably extension  of the inequality (\ref{e:inequality})
to arbitrary surfaces, as it is proven in the following theorem.

\begin{theorem}[Area, charge and global topology]
\label{t:main2} 
Let $(\Sigma,(h,K),(E,B))$  be a complete, maximal and asymptotically flat (with
possibly many asymptotic ends), initial data for Einstein-Maxwell equations. We
assume that the non-electromagnetic matter fields are non-charged and that they
satisfy the dominant energy condition.  Then for any oriented surface $\Su$
screening an end $\Sigma_{e}$ we have
\begin{equation}
\label{MI2AQ} 
A(S)\geq 4\pi(\bar Q_{E}^{2}+\bar Q_{M}^{2})\geq \frac{4\pi
  (Q^{2}_{E}+Q^{2}_{M})}{|H_{2}|}, 
\end{equation}
where  $Q_{E}$  and $Q_{M}$ are the electric and magnetic charges of
$\Su$,  $\bar Q_{E}$ and $\bar Q_{M}$ are the absolute central charges of $\Su$
and $H_2$ is the second Betti number of $\Sigma$.   
\end{theorem}
For the definitions of screening surface and absolute central charges see
section \ref{s:regions}. 
It is important to note that all the charges in theorem \ref{t:main2} are
produced by a non-trivial topology in the manifold (since by assumption the
non-electromagnetic fields are uncharged in the whole initial surface
$\Sigma$). That is, if the topology is trivial (i.e. $\Sigma=\R^3$) there is no
charges and the theorem is also trivial. This is an important difference with
theorem \ref{t:main1}, where the charge can be produced by charged matter
inside the trapped surfaces. Note also that this theorem has global
requirements (namely, asymptotic flatness, completeness and the assumption that
the matter is uncharged), in contrast with theorem \ref{t:main1} which is
purely quasi-local in the sense that only conditions at the 
surface are used. 

Let us discuss theorem \ref{t:main2} in some detail. In order to give an
intuitive idea of the result and of the definitions involved, in the
following we will analyze a particular class of examples.

Consider the well known Brill-Lindquist initial data \cite{Brill63}.
Brill-Lindquist data are time symmetric, conformally flat initial data with $N$
asymptotic ends. To simplify the discussion we take $N=3$ (in fact the
discussion below applies to a much general class of data which are not
necessarily conformally flat). The manifold is    
$\Sigma:=\R^3 \backslash \{x_1, x_2\}$, where $x_1$ and $x_2$ are arbitrary
points in $\R^3$. 
Let $L=|x_1-x_2|$ ,
where $|\cdot |$ denotes the Euclidean distance with respect to the flat
conformal metric. The end points
$x_1$ and $x_2$ have electric charges $Q_1$ and $Q_2$.  The other end has charge
$Q$ given by
\begin{equation}
  \label{eq:16}
  Q=Q_1+Q_2.
\end{equation}
Consider families of initial data with fixed charges but different separation
distance $L$. When $L$ is big enough, it can be proved that there exist only
two stable minimal surfaces $\Su_1$ and $\Su_2$ surrounding each end point.
See figure \ref{fig:1} (for a numerical picture of these surfaces see the
original article \cite{Brill63}, the analytical proof that there exist only
these two surfaces has been given in \cite{Chrusciel02a,Dain:2010pj}).
\begin{figure}
  \centering
  \includegraphics{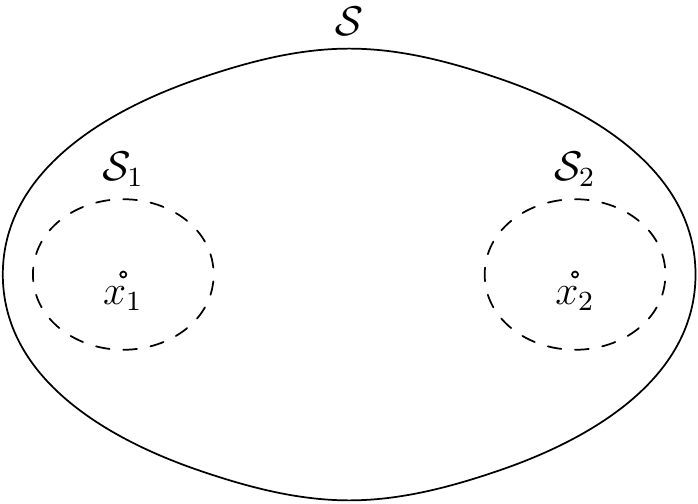}
  \caption{Brill-Lindquist data with large separation distance. The dashed
    surfaces are minimal surfaces.} 
  \label{fig:1}
\end{figure}
Take a sphere $\Su$ that encloses the two end points $x_1$ and $x_2$. This
surface is screening (for a precise definition see definition \ref{d:screening}
in section \ref{s:regions}).  Since $\Su_1$ and $\Su_2$ are the only minimal
surfaces, we have that
\begin{equation}
  \label{eq:5}
  A\geq A_1+A_2,
\end{equation}
where $A$ is the area of $\Su$ and $A_1$, $A_2$ are the areas of $\Su_1$ and
$\Su_2$ respectively. Applying theorem \ref{PQ2} for each minimal surfaces from
(\ref{eq:5}) we obtain
\begin{equation}
  \label{eq:32}
  A\geq 4\pi (Q_1^2+Q_2^2).
\end{equation}

\begin{figure}
  \centering
  \includegraphics{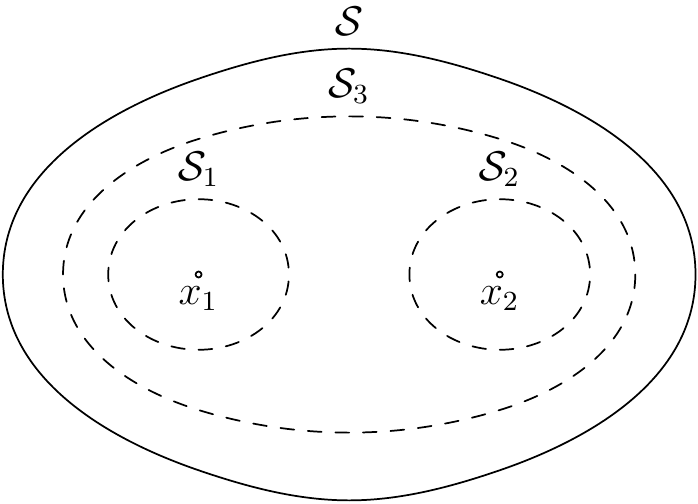}
  \caption{Brill-Lindquist data with small separation distance. A third minimal
    surface $\Su_3$ appears enclosing the two ends $x_1$ and $x_2$.}
  \label{fig:2}
\end{figure}
Take now $L$ to be small enough. Then a third minimal surface $\Su_3$, with
area $A_3$,  which
enclose the two ends appears. This surface is the outermost one and hence we
have 
\begin{equation}
  \label{eq:15}
  A\geq A_3. 
\end{equation}
 See figure
\ref{fig:2}. Then, using theorem  \ref{PQ2} we get 
\begin{equation}
  \label{eq:25}
     A(\Su)\geq 4\pi (Q_1+Q_2)^2=4\pi Q^2
\end{equation}
Where we have used that the charge of the surface $\Su_3$ is equal to the
charge of the end. If we combine inequality (\ref{eq:25}) with (\ref{eq:32}) we
obtain the following
\begin{equation}
  \label{eq:26}
  A(\Su)\geq 4\pi \inf\{Q^2_1+Q^2_2, (Q_1+Q_2)^2\}. 
\end{equation}
This inequality is valid for all screening surfaces $\Su$ and it is independent
of $L$.  The right hand side of this inequality is precisely the square of the
absolute central charge defined in section \ref{s:regions}, namely
\begin{equation}
  \label{eq:27}
  \bar Q(\Su)=\sqrt{\inf\{Q^2_1+Q^2_2, (Q_1+Q_2)^2\}}.
\end{equation}
Note that if $Q_1$ and $Q_2$ have opposite signs, we get
\begin{equation}
  \label{eq:29}
  \bar Q(\Su)=|Q_1+Q_2|,
\end{equation}
and if they have the same signs we get
\begin{equation}
  \label{eq:30}
  \bar Q(\Su)=\sqrt{Q^2_1+Q^2_2}.
\end{equation}
The Betti number $H_2$ measures the number of holes of $\Su$, in the present
case we have  $H_2=2$. It is clear that
\begin{equation}
  \label{eq:28}
  \bar Q(\Su) \geq \frac{|Q_1+Q_2|}{2}= \frac{Q(\Su)}{H_2}.
\end{equation}
This is precisely the second inequality in (\ref{MI2AQ}).  Note that knowing
the size of the parameter $L$ provide finer information.  For example, take
$Q_1=-Q_2$. In that case $Q(\Su)=\bar Q(\Su)=0$ and theorem \ref{t:main2} is
trivial. However, if $L$ is big, we have the non-trivial inequality
(\ref{eq:32}).

Finally we present our third main result. As we discussed above theorem
\ref{t:main2} generalizes theorem \ref{s:main} in the sense that it applies to
surfaces that are not necessarily black holes horizons.  However, in that
theorem a strong restriction is made, namely that matter fields have
no charges. The natural question is what happens for an ordinary charged
object, is it possible to prove a similar kind of inequality? The answer is no.
There exists an interesting and highly non-trivial counter example.  This
counter example was constructed by W. Bonnor in \cite{bonnor98} and it can be
summarized as follows: \emph{for any given positive number $k$, there exist 
static, isolated, non-singular
  bodies, satisfying the energy conditions, whose surface area $A$ satisfies $A
  < kQ^2$.}  In the
article \cite{bonnor98} the inequality is written in terms of the mass, however
for this class of solution the mass is always equal to the charge of the
body. The body is a highly prolated spheroid of electrically counterpoised
dust. This suggests that for a body which is `round' enough a version of
inequality (\ref{MI2AQ}) can still holds. From the physical point of view we
are saying that for an ordinary charged object we need to control another
parameter (the `roundness') in order to obtain an inequality between area and
charge. Remarkably enough it is possible to encode this intuition in the
geometrical concept of isoperimetric surface: we say that a surface $\Su$ is
isoperimetric if among all surfaces that enclose the same volume as $\Su$ does,
$\Su$ has the least area. Then using the same technique as in the proof of
theorem \ref{PQ2}  and applying the results of \cite{christodoulou88} we obtain the
following theorem for isoperimetric surfaces. 

\begin{theorem}
\label{t:main3} 
Consider an electro-vacuum, maximal initial data, with a non-negative
cosmological constant. Assume that $\Su$ is a  stable isoperimetric
sphere. Then
\begin{equation}
  \label{eq:isoparametric} 
A(\Su)\geq \frac{4\pi}{3}(Q^{2}_{E}+Q^{2}_{M}), 
\end{equation}
where  $Q_{E}$  and $Q_{M}$ are the electric and magnetic charges of
$\Su$.   
\end{theorem}
We emphasize that this theorem is purely quasi-local (as theorem \ref{s:main}),
it only involves conditions on the surface $\Su$. In particular, it is assumed
electro-vacuum only on $\Su$, charged matter could exist inside or outside the
surface.

\section{Area--charge inequality for black holes}
\label{s:black-holes}

The aim of this section is to prove theorem \ref{t:main1}.  We follow the
notation and definitions presented in \cite{Jaramillo:2011pg}.  Consider a
closed orientable 2-surface ${\cal S}$ embedded in a spacetime $M$ with metric
$g_{ab}$ and Levi-Civita connection $\nabla_a$. We denote the induced metric
on ${\cal S}$ as $q_{ab}$, with Levi-Civita connection $D_a$ and Ricci scalar
${}^2\!R$. We will denote by $dS$ the area
measure on ${\cal S}$. Let us consider null vectors $\ell^a$ and $k^a$
spanning the normal plane to ${\cal S}$ and normalized as $\ell^a k_a = -1$,
leaving a (boost) rescaling freedom $\ell'^a =f \ell^a$, $k'^a = f^{-1} k^a$.
The expansion $\theta^{(\ell)}$ and the shear $\sigma^{(\ell)}_{ab}$ associated
with the null normal $\ell^a$ are given by 
\begin{equation}
\label{e:expansion_shear}
\theta^{(\ell)}=q^{ab}\nabla_a\ell_b, \quad  
\sigma^{(\ell)}_{ab}=  {q^c}_a {q^d}_b \nabla_c \ell_d - \frac{1}{2}\theta^{(\ell)}q_{ab} \ ,
\end{equation}
whereas the normal fundamental form $\Omega_a^{(\ell)}$ is 
\begin{equation}
\label{e:Omega}
\Omega^{(\ell)}_a = -k^c {q^d}_a \nabla_d \ell_c \ .
\end{equation}
The spacetime metric $g_{ab}$ can be written in the following form
\begin{equation}
  \label{eq:6}
  g_{ab}=q_{ab} -\ell_a k_b - \ell_b k_a \ .
\end{equation}
The surface ${\cal S}$ is a marginal outer trapped surface if $\theta^{(\ell)}=0$. We
will refer to $\ell^a$ as the {\em outgoing} null vector.

The following stability condition on marginally trapped surfaces 
introduced in Refs.\cite{Andersson:2005gq,andersson08}, plays a
crucial role.

\begin{definition}
\label{d:stability_AMS}
(Andersson, Mars, Simon)
Given a closed marginally trapped surface ${\cal S}$
and a vector $v^a$ orthogonal to it, 
we will refer to ${\cal S}$ as {\em stably outermost with 
respect to the direction $v^a$} iff there exists a function
$\psi>0$ on ${\cal S}$ such that the variation
of $\theta^{(\ell)}$ with respect to $\psi v^a$
fulfills the condition
\begin{equation}
\label{e:stability_condition}
\delta_{\psi v} \theta^{(\ell)} \geq 0.
\end{equation}
\end{definition}
Here $\delta$ denotes the variation operator associated with a 
deformation of the surface ${\cal S}$ introduced in \cite{Andersson:2005gq} 
(see also the treatment in \cite{Booth:2006bn}).
Following \cite{Jaramillo:2011pg} we will formulate 
this stability notion in a sense not referring
to a particular stability direction, but just requiring stability
along some outgoing non-timelike direction.

\begin{definition}
\label{d:stability}
A closed marginally trapped surface ${\cal S}$ is referred
to as {\em spacetime stably outermost} if there exists 
an outgoing ($-k^a$-oriented)
vector $x^a= \bar{\gamma} \ell^a - k^a$, with 
$\bar{\gamma}\geq0$, with respect to which  ${\cal S}$ is stably outermost.
\end{definition}
In the following, we denote by $X^a$ the vector $X^a = \psi x^a=
\gamma \ell^a - \psi k^a$, with $\psi$ the function guaranteed
by Definition \ref{d:stability_AMS} and $\gamma \equiv\psi \bar{\gamma}$,
so that $\delta_X \theta^{(\ell)} \geq 0$.
Note that this {\em spacetime stability condition}
includes, for  an outgoing past null vector $x^a=-k^a$, 
the (outer trapping horizon) stability notions
in \cite{hayward94,Racz:2008tf}.
For further discussion concerning this stability condition 
see \cite{Jaramillo:2011pg}.

The following Lemma provides the essential estimate for the matter fields on a
stable marginally trapped surface $\Su$. It is the analog of Lemma 1 in
\cite{Jaramillo:2011pg}. Its proof essentially follows from setting the function
$\alpha=1$ used in that Lemma. It is important to emphasize that no symmetry
assumption is made.  For completeness and since the final proof is much simpler
we present it here.
\begin{lemma}
  Given a closed marginally trapped surface ${\cal S}$ satisfying the spacetime
  stably outermost condition then the following inequality holds
\begin{equation}
\label{e:inequality_alpha}
 \int_{\cal S} \left[ G_{ab}\ell^a (  k^b + \frac{\gamma}{\psi} \ell^b) \right]
 dS \leq  4\pi (1-g),  
\end{equation}
where $g$ is the genus of $\Su$. If in addition we assume that the left hand
side in the inequality (\ref{e:inequality_alpha}) is non-negative and not
identically zero, then it follows that $g=0$ and hence $\Su$ has the
$\mathbb{S}^2$ topology.
\end{lemma}
\begin{proof}
    First, we evaluate 
$\delta_X \theta^{(\ell)}/\psi$  for the vector 
$X^a=\gamma \ell^a - \psi k^a$ provided by Definition 1,
(use e.g. Eqs. (2.23) and (2.24) in  \cite{Booth:2006bn})
and impose $\theta^{(\ell)}=0$. We obtain
\begin{multline}
\label{e:delta_X_theta}
\frac{1}{\psi}\delta_X\theta^{(\ell)} = D^a  \Omega^{(\ell)}_a - {}^2\!\Delta \mathrm{ln}\psi
 -  D_a\mathrm{ln}\psi D^a\mathrm{ln}\psi
+ 2 \Omega^{(\ell)}_a  D^a\mathrm{ln}\psi - \Omega^{(\ell)}_c  {\Omega^{(\ell)}}^c \\
 +\frac{1}{2}{}^2\!R - \frac{\gamma}{\psi} \left[\sigma^{(\ell)}_{ab} {\sigma^{(\ell)}}^{ab} 
+ G_{ab}\ell^a\ell^b \right] \nn  - G_{ab}k^a\ell^b   .
\end{multline}
We integrate this equation over the surface $\Su$. On the  left hand side we use
the stability condition (\ref{e:stability_condition}). The first two  terms
 in the right hand side integrate to zero. 
The next three terms can be arranged as a total square, namely
\begin{equation}
  \label{eq:1}
-(D_a\mathrm{ln}\psi- \Omega^{(\ell)}_a )(D^a\mathrm{ln}\psi-
{\Omega^{(\ell)}}^a )=  
- D_a\mathrm{ln}\psi D^a\mathrm{ln}\psi
+ 2 \Omega^{(\ell)}_a  D^a\mathrm{ln}\psi - \Omega^{(\ell)}_c  {\Omega^{(\ell)}}^c,
\end{equation}
and hence the integral is non-positive. The integral of the scalar curvature is
calculated using the Gauss-Bonnet theorem
\begin{equation}
  \label{eq:2}
  \int_S \frac{1}{2}{}^2\!R dS  =4\pi(1-g) .
\end{equation}
Finally, the term with $\sigma^{(\ell)}_{ab} {\sigma^{(\ell)}}^{ab} $ is
non-positive. Collecting all these observations, the inequality
(\ref{e:inequality_alpha}) follows. If the left hand side of the inequality
(\ref{e:inequality_alpha}) is non-negative it follows that $g$ can be $0$ or
$1$. If it is not identically zero then $g=0$ and hence $\Su$ has the
$\mathbb{S}^2$ topology.
\end{proof}

The following lemma will allow us to write the relevant normal components of
the electromagnetic field on the surface in terms of the charges. It is
important to note that it is a pure algebraic result, Maxwell equations are not
used. In particular, the generalization to Yang-Mills theories with a compact
Lie group is direct and will be presented elsewhere.
\begin{lemma}
\label{e:Tkl_EM}
  Let $T^{EM}_{ab}$ be the electromagnetic energy-momentum tensor given by
  (\ref{eq:4}). Then the following equality holds
  \begin{equation}
    \label{eq:7}
    T^{EM}_{ab}\ell^ak^b = \frac{1}{8\pi} 
    \left[\left(\ell^a k^b F_{ab}\right)^2 
    +\left(\ell^a k^b {}^*\!F_{ab}\right)^2\right].
  \end{equation}
\end{lemma}
\begin{proof}
  The proof is a straightforward  computation using the form of the
  metric (\ref{eq:6}). We mention some useful intermediate steps. Using
  equation (\ref{eq:6}) we calculate
  \begin{equation}
    \label{eq:8}
    F_{ab}F^{ab}=-2\left( \ell^a k^b F_{ab} \right)^2-4q^{ab} k^c F_{ac} \ell^d
    F_{bd}+ F_{ab} F_{cd}q^{ac}q^{bd},
  \end{equation}
and
\begin{equation}
  \label{eq:9}
  \ell^a  k^c F_{ab} F_{c}{}^b= \left( \ell^a k^b F_{ab} \right)^2 + q^{ab} k^c
  F_{ac} \ell^d F_{bd}.
\end{equation}
Noting that the pull-back of $F_{ab}$ on the surface $\Su$ is proportional to
the volume element $\epsilon_{ab}$ of the surface $\Su$, we can evaluate
$F_{ab} F_{cd}q^{ac}q^{bd}$ and $\left(\epsilon^{ab}F_{ab}\right)^2$ to obtain
\begin{equation}
  \label{eq:9.1}
   F_{ab} F_{cd}q^{ac}q^{bd} = \frac{1}{2} \left(\epsilon^{ab}F_{ab}\right)^2
   = 2 \left({}^*\!F_{ab}\ell^ak^b\right)^2 \ ,
\end{equation}
where the following identity 
\begin{equation}
  \label{eq:5.2}
{}^*\!F_{ab}\ell^ak^b = \frac{1}{2} F_{ab}\epsilon^{ab},
\end{equation}
has been used in the second equality. This identity follows from the relation
$\epsilon_{ab}= \epsilon_{abcd}\ell^c k^d$.  Inserting first (\ref{eq:9.1}) in
Eq. (\ref{eq:8}) and then, the resulting expression [together with
(\ref{eq:9})] into (\ref{eq:4}), we obtain (\ref{eq:7}).
\end{proof}
Note that the electric and magnetic charges (\ref{eq:4.5}) of $\Su$ can be written as
follows in terms of the null vector $\ell^a$ and $k^a$
\begin{equation}
  \label{eq:5.1}
  Q_E=\frac{1}{4\pi}\int_{\Su} F_{ab}\ell^ak^b dS, \quad  
  Q_M = \frac{1}{4\pi}\int_{\Su} {}^*\!F_{ab}\ell^ak^b dS. 
\end{equation}

Having proved these two lemma we have already the basic ingredients for 
the proof of our first main result. 

\begin{proof}[Proof of theorem \ref{t:main1}]
  We use inequality (\ref{e:inequality_alpha}) and Einstein equations
(\ref{eq:3a}).  Since the vector $k^a+\gamma/\psi \ell^a$ is timelike or null,
  using that the tensor $T_{ab}$ satisfies the dominant energy condition and
  that $\Lambda$ is non-negative we get from (\ref{e:inequality_alpha})   that
\begin{equation}
  \label{eq:10}
  8\pi  \int_{\cal S}  T^{EM}_{ab}\ell^a k^b   dS  \leq 8\pi
  \int_{\cal S} \left[ T^{EM}_{ab}\ell^a (  k^b + \frac{\gamma}{\psi} \ell^b)
  \right] dS \leq  4\pi(1-g). 
\end{equation}
where in the last inequality we have used that $T^{EM}_{ab}\ell^a \ell^b\geq 0$
(this inequality follows directly from (\ref{eq:4}), i.e.  the electromagnetic
energy-momentum tensor satisfies the null energy condition).  We use equality
(\ref{eq:7}) to obtain from inequality (\ref{eq:10}) the following bound
\begin{equation}
  \label{eq:11}
\int_{\cal S}   \left[\left(\ell^a k^b F_{ab}\right)^2 
    +\left(\ell^a k^b {}^*\!F_{ab}\right)^2\right]dS \leq 4\pi(1-g).
\end{equation}
If the left hand side of the inequality (\ref{eq:11}) is identically zero then
the charges are zero and the inequality (\ref{e:inequality}) is trivial. Then
we can assume that it is not zero at some point and hence we have that $g=0$.

To bound the left hand side of inequality (\ref{eq:11}) we use H\"older
inequality on $\Su$ (following the spirit of the proof presented in
\cite{Jang79} for the charged Penrose inequality) in the following form. For
integrable functions $f$ and $h$, H\"older inequality is given by
\begin{equation}
  \label{eq:12}
  \int_{\Su}  fh dS\leq  \left(\int_{\cal S}  f^2 dS\right)^{1/2}
  \left(\int_{\Su}  h^2 dS\right)^{1/2}.
\end{equation}
If we take $h=1$, then we obtain
\begin{equation}
  \label{eq:13}
  \int_{\Su}  f dS\leq  \left(\int_{\Su}  f^2 dS\right)^{1/2}  A^{1/2}.
\end{equation}
where $A$ is the area of $\Su$. 
Using this inequality in (\ref{eq:11}) we finally obtain
\begin{equation}
  \label{eq:14}
  A^{-1}\left[ \left (\int_{\cal S} \ell^a k^b F_{ab} dS \right)^2
      + \left (\int_{\cal S} \ell^a k^b {}^*\!F_{ab} dS \right)^2 \right]
\leq 4\pi. 
\end{equation}
Finally, we use Eq. (\ref{eq:5.1}) to express the
 left-hand-side of (\ref{eq:14})  in terms of $Q_E$ and  $Q_M$.
Hence the inequality (\ref{e:inequality}) follows. 
\end{proof}
We note that, up to the use of H\"older inequality in Eq. (\ref{eq:12}), the
line of reasoning in the proof above is also followed in \cite{Booth:2007wu}.
Starting from the {\em outer} condition for trapping horizons in  \cite{hayward94}
(see also \cite{Racz:2008tf}), namely the stably outermost condition for a null
$X^a$, a version of Lemma 3.2 is derived there [their Eq. (20)].  Then, the
equality in Lemma 3.3 is their Eq. (22). The last step completing the proof is
though missing.

\section{Area, charge and global topology.}
\label{s:regions}

We consider maximal Einstein-Maxwell initial states $(\Sigma,(h,K),(E,B))$,
with possibly many asymptotically flat (AF) ends.  Asymptotically flat ends
will be denoted by $\Sigma_{e}$. Our central object, the subject of our study,
will be surfaces, $\Su$, ``{\it screening}" a given end $\Sigma_{e}$. Their
definition is as follows.

\begin{definition}[Screening surfaces]
\label{d:screening}
  Fix an AF end $\Sigma_{e}$ of $\Sigma$. A compact, oriented, but not
  necessarily connected surface $\Su$ is said to screen the end $\Sigma_{e}$ if
  it is the boundary of an open and connected region $\Omega$ containing the
  given end but not any other. Such $\Omega$ is called a screened region.
\end{definition}

Every component of the surface $\Su$ will be given always the orientation arising
from the outgoing normal to $\Omega$.

Given an embedded oriented and compact surface $S$, and a divergenceless vector
field $X^a$ we define the charge $Q(\Su)$ (relative to $X^a$) as 
\begin{equation}
\label{e:charge}
Q(\Su)=\frac{1}{4\pi}\int_{\Su} X_a n^a dS, 
\end{equation}
where $n^a$ is the normal field to $\Su$ in $\Sigma$, that, 
together with the orientation of
$\Sigma$ returns the orientation of $\Su$. Note that because $X^a$ is
divergenceless, the charge $Q(\Su)$ depends only on the homology class of $\Su$,
denoted by $[\Su]$. When $X^a=E^a$ or $X=B$, that is, when $X^a$ is either the electric
or the magnetic field, then the associated charges are the electric or the
magnetic charges. To avoid excessive writing and to display certain generality,
we will work most of times with an arbitrary vector field $X^a$, instead of the
specific vectors $E^a$ and $B^a$.

Note, by the Gauss Theorem, that if $\Su$ is screening then the electric or the
magnetic charges of $\Su$ are equal to the electric or the magnetic charges of
the given end $\Sigma_{e}$.

In the following we will discuss the notion of {\it absolute central charges}
associated to an end which will play an important role in the proof of Theorem
\ref{t:main2}. The relevant properties of charges and central charges are
summarized in Proposition \ref{PQ1}. Then we will explain in Proposition
\ref{PQ2} the basic inequality between area and charge for stable minimal
surfaces. Using these elements we sketch then the idea of the proof of Theorem
\ref{t:main2}. The rigorous proof is given immediately thereafter.

\begin{definition}[Absolute central charges]
  Fix an AF end $\Sigma_{e}$. Let $\Su=\Su_{1}\cup\ldots\cup \Su_{k(\Omega)}=\partial
  \Omega$ be a screening surface of the end $\Sigma_{e}$. Among the $\Su_{i}$'s
  there are those that are part of the boundary of the unbounded connected
  components of $\Sigma\setminus \Omega$. Let us assume that
  $\{\Su_{1},\ldots,\Su_{k(\Omega)}\}$ were ordered in such a way that
  $\{\Su_{1},\ldots,\Su_{n(\Omega)}\}$, $n(\Omega)\leq k(\Omega)$, are such
  components. Then, define the Absolute Central Electric or Magnetic Charges
  $\bar Q_{E}$ and $\bar Q_{M}$ associated to an end $\Sigma_{e}$ as
\begin{equation}
\label{FACCH}
\bar Q=\inf_{\Omega}\sqrt{\sum_{i=1}^{i=n(\Omega)}Q^{2}(\Su_{i})},
\end{equation}
where for $\bar Q_{E}$, $Q$ is the electric charge and for $\bar Q_{M}$, $Q$ is the
magnetic charge and where $\Omega$ ranges among the screened regions of
$\Sigma_{e}$.
\end{definition}

We note now some basic facts about charges and absolute central charges. 
\begin{proposition}\label{PQ1}
  Let $\Omega$ be a screened region of an end $\Sigma_{e}$, and let $\Su=\partial
  \Omega$ be the screening surface. Then
\begin{enumerate}

\item $|Q(\Sigma_{e})|=|\sum_{i=1}^{i=n(\Omega)} Q(\Su_{i})|\leq
  n(\Omega)^{\frac{1}{2}}(\sum_{i=1}^{i=n(\Omega)} Q^{2}(\Su_{i}))^{\frac{1}{2}}$,
  where $Q$ is here either an electric or a magnetic charge.

\item $n(\Omega)\leq |H_{2}|$, where $|H_{2}|$ is the second Betti
  number\footnote{Recall that the second homology group
    $H_{2}(\Sigma,\field{Z})$ of a manifold $\Sigma$ (with finitely many AF
    ends) is always of the form $H_{2}\sim \field{Z}^{|H_{2}|}\oplus T$, where
    $T$ is a finite abelian group called the {\it Torsion} and where $|H_{2}|$
    is the second Betti number.}.

\item $Q^{2}(\Sigma_{e})/|H_{2}|\leq \bar Q^{2}(\Su)$.
\end{enumerate}
\end{proposition}

\begin{proof}

{\it Item 1}. Let $\Sigma\setminus
  \Omega=\cup_{i=1}^{i=j(\Omega)}\Omega^{c}_{i}$, where the $\Omega^{c}_{i}$
  are connected. Then we have 
\begin{equation} 
Q(\Sigma_{e})=\sum_{i=1}^{i=j(\Omega)}
  Q(\partial \Omega^{c}_{i}), 
\end{equation}
where $Q$ is either the electric or the magnetic charge. But we note that if
$\Omega^{c}_{i}$ is a bounded component then by the Gauss Theorem $Q(\partial
\Omega^{c}_{i})=0$. Using this and recalling then that the surfaces
$\Su_{1},\ldots,\Su_{n(\Omega)}$ are those that belong to the boundary of an
unbounded connected component, $\Omega^{c}_{i}$, of $\Sigma\setminus \Omega$ we
obtain 
\begin{equation} 
\bar Q(\Sigma_{e})=\sum_{i=1}^{i=j(\Omega)} Q(\partial
\Omega^{c}_{i})=\sum_{i=1}^{i=n(\Omega)} Q(\Su_{i}), 
\end{equation}
and the claim of the {\it item 1} follows.

\vs 

{\it Item 2}. We show now that the surfaces $\Su_{1},\ldots,\Su_{n(\Omega)}$,
which are orientable and oriented (from the outgoing normal to $\Omega$), are
indeed linearly independent in $H_{2}(\Sigma,\field{Z})$. Namely we show that
if for some integer coefficients $a_{i}\in \field{Z}$, $i=1,\ldots,n(\Omega)$,
we have
\begin{equation}
\sum_{i=1}^{i=n(\Omega)} a_{i}[\Su_{i}]=0,
\end{equation}
in $H_{2}$ then $a_{i}=0$ for $i=1,\ldots, n(\Omega)$. Thus
$\field{Z}^{n(\Omega)}\subset H_{2}$ and therefore $n(\Omega)\leq |H_{2}|$.

A simple and visual way to show this using triangulations of $\Sigma$ is the
following. 

Suppose that a certain integer combination of $[\Su_{i}]$ is zero in Homology,
namely suppose that $\sum_{i=1}^{i=n} a_{i}[\Su_{i}]=\partial [\bar{C}_{3}]$, for
some integer coefficients $a_{i}$ and a singular chain $[\bar{C}_{3}]=\sum
b_{i} [\sigma_{i}]$, where $\sigma_{i}:\Delta^{3}\rightarrow \Sigma$ is a
singular three-simplex (\cite{Hatcher02}, pg. 108). We consider now a closed
region $\bar{\Sigma}$, with smooth boundary and containing in its interior the
surfaces $\Su_i$ and the singular simplices $\sigma_{i}(\Delta^{3})$. It is
clear that $\sum_{i=1}^{i=n}a_{i}[\Su_{i}]=0$ in $H_{2}(\bar{\Sigma},\field{Z})$.

Consider a triangulation of $\Sigma$ by embedded three-simplices
(i.e. tetrahedrons) in such a way that every embedded two-simplex
(i.e. triangle) of their boundaries is either disjoint from all the $\Su_i$'s
and $\partial \bar{\Sigma}$ or is inside and embedded in one of the $\Su_i$'s
or in $\partial \bar{\Sigma}$ (such triangulation always exists). We are going
to think in this way $\bar{\Sigma}$ as a $\Delta$-complex (\cite{Hatcher02},
pg. 104).

We recall that the homology groups of $\bar{\Sigma}$ as a $\Delta$-complex,
denoted by $H^{\Delta}_{i}(\Sigma,\field{Z}),i=0,1,2,3$ and the homology groups
of $\bar{\Sigma}$, denoted by $H_{i}(\bar{\Sigma},\field{Z})$, $i=0,1,2,3$ are
naturally isomorphic (\cite{Hatcher02}, Thm. 2.27)

For this reason it is enough to argue in terms of chains
of the $\Delta$-complex (triangulation) only. We will do that in the following.

Note that for the particular triangulation that we have chosen we can think
$[\Su_{i}]$ as a two-chain of the $\Delta$-complex, namely a sum with
coefficients in $\field{Z}$ of oriented three-simplices of the
$\Delta$-complex. The same happens with $\partial \bar{\Sigma}$. Suppose then
that $\sum_{i=1}^{i=n(\Omega)}a_{i}[\Su_{i}]=0$ in $H_{2}^{\Delta}$, that is,
suppose that
\begin{displaymath}
\sum_{i=1}^{i=n}a_{i}[\Su_{i}]=\partial [C_{3}],
\end{displaymath}
where $a_{i}\in \field{Z}$, and $[C_{3}]$ is a three-chain of the
$\Delta$-complex, namely a sum with coefficients in $\field{Z}$ of oriented
three-simpllices of the $\Delta$-complex.  We want to see that all the
$a_{i}'s$ must be zero. For this we will make use of smooth embedded,
inextensible, oriented curves, denoted by $\xi$, such that
\begin{enumerate}
\item $\xi$ ends along one direction at $\Sigma_{e}$ and ends along the other
  direction at another end $\Sigma_{e}'$, ($\Sigma_{e}'\neq \Sigma_{e}$).

\item if $\xi$ intersects a two-simplex of the $\Delta$-complex it does so in its
interior and transversally to it. Thus, if $\xi$ intersects $\Su_i$ then it
does so transversally.

\end{enumerate}

Thus, because $\xi$ and $\Su_i$ are oriented, their intersection number
(\cite{Guillemin-Pollack}, Ch. 3, $\S\ 3$) denoted by $\varhash(\xi\cap \Su_i)$
is well defined\footnote{\cite{Guillemin-Pollack} uses this notation for the
  intersection number $mod\ 2$}.  Moreover we have \be\label{INTN}
\sum_{i=1}^{i=n}a_{i}\varhash(\xi \cap [\Su_i])=\varhash(\xi\cap \partial
[C_{3}]), \ee
 
We note now that the boundary of any three-simplex of the $\Delta$-complex has
signed intersection number equal to zero to any such curve ($\xi$ gets out of
the three-simplex the same number of times it gets in). Therefore the
intersection number of any curve $\xi$ with $\partial[C_{3}]$ must be
zero. Therefore from (\ref{INTN}) we get 
\begin{equation}
\label{INTEN2}
\sum_{i=1}^{i=n}a_{i}\varhash(\xi \cap [\Su_i])=0.  
\end{equation}
for any such curve $\xi$. Assume now that $a_{j}\neq 0$. Recalling the
definition of the $\Su_i's$ we can consider an inextendible curve $\xi$ as
before, such that $\varhash(\xi\cap S_{j})=1$ and $\varhash(\xi\cap \Su_i)=0$
for $i\neq j$. Indeed the curve $\xi$ can be chosen to intersect $S_{j}$ only
once and avoiding intersecting $\Su_i,\ i\neq j$.  Then, the intersection
number of $\xi$ to $\sum a_{i}[\Su_i]$, must be equal to
\begin{equation}
\sum_{i=1}^{i=n(\Omega)}a_{i} \varhash(\xi\cap [\Su_i])=a_{j}\neq 0,
\end{equation}
which is a contradiction. This finishes the proof of the second {\it item}.

\vs
{\it Item 3}. This {\it item} follows directly from {\it items 1} and {\it 2}.

\end{proof}

We discuss now the basic relation between charge and area for stable minimal
surfaces. We recall first the setup.  Let $(\Sigma,h)$ be an oriented Riemannian
three-manifold, with possibly many asymptotically flat ends. Suppose that its
scalar curvature $R$ satisfies $R\geq 2 |X|^2$, where the vector field $X^a$ is
divergence-less. Then for any oriented surface $\Su$, the charge $Q([\Su])$ is
given by \eqref{e:charge}.  Then, in this setup, we have the following result
proved in \cite{Gibbons:1998zr}. For completeness we repeat its proof.
\begin{theorem}[Gibbons]
\label{PQ2}
Let $\Su$ be a stable minimal surface. Then
\begin{equation}
\label{MIAQ}
A\geq 4\pi Q^{2},
\end{equation}
where $A$ is the area of $\Su$ and $Q$ is its charge.
\end{theorem}

\begin{proof}

The stability inequality (where $D$ is the covariant derivative with respect to
the Riemannian metric $h$)
\begin{equation}
\label{e:sta-min}  
\int_{S}|D \alpha|^{2}+\frac{1}{2} {}^2\!R    \alpha^{2}\, dS\geq
\frac{1}{2}\int_{S} R \, dS. 
\end{equation}
with $\alpha=1$ gives 
\begin{equation}
4\pi \geq \frac{1}{2}\int_{S} RdS\geq \int_{S}|X|^{2}dS\geq
\frac{(\int_{S} X_a n^a dS)^{2}}{A}=\frac{(4\pi Q)^{2}}{A},
\end{equation}
where the last inequality follows from the Cauchy-Schwarz inequality.
\end{proof}

\vs
Note that part of the argument above shows that
\begin{equation}
A\leq \frac{4\pi}{\overline{|X|^{2}}}.
\end{equation}
where $\overline{|X|^{2}}$ is the average of $|X|^{2}$ over $S$. Combining this
and (\ref{MIAQ}) in the case of the electromagnetic field (Einstein-Maxwell) we
get 
\begin{equation} 
\overline{|E|^{2}+|B|^{2}}\leq \frac{1}{Q_{E}^{2}+Q_{M}^{2}}.  
\end{equation} 
{\it
  In other words, the average of the electromagnetic energy over $\Su$ is bounded
  above by the sum of the squares of the electric and magnetic charges.} In a
mean-sense, the electromagnetic energy cannot be arbitrarily large over $\Su$ if
$S$ is minimal and stable.

\vs We are ready to discuss and give the proof of Theorem \ref{t:main2}. As
said before and to simplify the writing we will work with a system of the form
\begin{align*}
&R\geq 2|X|^{2},\\
& D_a X^a=0,
\end{align*}
\n instead of the system
\begin{align}
&R\geq 2(|E|^{2}+|B|^{2}),\\
&D_a E^a=0,\\
&D_a B^a=0,
\end{align}

\n but the argumentation is exactly parallel in this last case. 

In this setup, the proof of Theorem \ref{t:main2} follows from Propositions
(\ref{PQ1}) and (\ref{PQ2}) and an application of a result of Meeks-Simon-Yau
\cite{Meeks82}.  Indeed, we start by choosing an end $\Sigma_{e}$ and a
screening surface $\Su$.  We apply then Theorem 1 in \cite{Meeks82} to obtain a
smooth measure-theoretical limit of isotopic variations of $\Su$, whose area
realizes the infimum of the areas of all the isotopic variations of $\Su$. The
important fact is that, because $\Su$ is screening, and the limit surfaces
(possibly repeated) are a measure-theoretical limit of isotopic variations of
$S$, then there is a subset of connected limit surfaces whose union is a
screening surface. The inequality (\ref{MI2AQ}) follows then applying
(\ref{MIAQ}) to any one of these stable components of the limit and using {\it
  Item 3} in Proposition \ref{PQ1}.

\vs
\n 
\begin{proof}[Proof of Theorem \ref{t:main2}]

  \vs Let $\Su$ be an oriented surface embedded in $\Sigma$ and screening the end
  $\Sigma_{e}$. Following \cite{Meeks82}, Theorem 1, there exist embedded minimal surfaces, $S_{1},\ldots,S_{k}$, and natural numbers $n_{1},\ldots,n_{k}$ ($n_{i}\geq 0$) such that 

\begin{enumerate}
\item
$A(\Su)\geq \inf_{\tilde{\Su}\sim \Su}
  A(\tilde{\Su})=n_{1}A(\Su_{1})+\ldots+n_{k}A(\Su_{k})$, where $\tilde{\Su}\sim \Su$
  signifies that the infimum is taken over surfaces $\tilde{S}$ isotopic to
  $S$, and,
\item there is a sequence of surfaces $\{\tilde{\Su}\}$ isotopic to $\Su$ such that
  for any continuous function $h$ we have
\begin{equation} 
\lim\int_{\tilde{\Su}} hdS=\sum_{i=1}^{i=k} n_{i}\int_{\Su_i}hdS.  
\end{equation}
which implies, choosing $h=1$, that $\lim A(\tilde{\Su})=n_{1}A(\Su_{1})+\ldots+n_{k}A(\Su_{k})$.
\end{enumerate} 
 
\n We claim that, because $\Su$ screens the end $\Sigma_{e}$, then there is a
subset of surfaces $\Su_{1},\ldots,\Su_{k}$ screening $\Sigma_{e}$. Namely we claim
that there is a screened region $\bar{\Omega}$, such that $\partial
\bar{\Omega}$ is a union of some or all of the surfaces
$\Su_{1},\ldots,\Su_{k}$. Let us postpone this technical point to the end, and
assume for the moment that the surfaces $\Su_i$'s were ordered in such a way
that $\Su_{1},\ldots \Su_{l}$, $l\leq k$ is such set of oriented surfaces, or in other words
that $\partial \bar{\Omega}=\Su_{1}\cup\ldots\cup \Su_{l}$.

We therefore calculate 
\begin{align}
  A(S)&\geq \sum_{i=1}^{i=k}n_{i}A(\Su_i)\geq 4\pi\sum_{i=1}^{i=l} n_{j}Q^{2}(\Su_i)\\
  &\geq 4\pi \sum_{i=1}^{i=l} Q^{2}(\Su_i)\geq 4\pi Q^{2}\geq \frac{4\pi Q^{2}}{|H_{2}|}.
\end{align}
The claim of Theorem \ref{t:main2} follows.

We prove now that there is a subset of the $\Su_{1},\ldots,\Su_{k}$ screening
$\Sigma_{e}$. For this we will show that every embedded inextensible curve
$\xi$ starting at $\Sigma_{e}$ and ending at $\Sigma_{e}\neq \Sigma_{e}$ has to
intersect one of the $\Su_{1},\ldots,\Su_{k}$. If that is the case define $\Omega$
as the set of points $p$ in $\Sigma\setminus (\Su_{1}\cup \ldots \cup \Su_{k})$,
such that there is an inextensible embedded curve $\beta$ starting at
$\Sigma_{e}$ and ending at $p$ and not touching any of the surfaces
$\Su_{1},\ldots,\Su_{k}$. Such open set would not contain any end different from
$\Sigma_{e}$ and its boundary would be a subset of $\Su_{1},\ldots,\Su_{k}$. Then
the closure $\bar{\Omega}$ of $\Omega$ must be a screened region and its
boundary $\partial \bar{\Omega}$ must be a subset of the
$\Su_{1},\ldots,\Su_{j}$. Note that $\partial \bar{\Omega}$ is not necessarily
equal to $\partial \Omega$.

Suppose now that there is an inextensible embedded curve $\xi$ starting at $\Sigma_{e}$
and ending at $\Sigma_{e}'\neq \Sigma_{e}$. 

Let now $T(r)$, for $r$ small, be a tubular neighborhood of $\xi$ of radius $r$ such that
$T(r)\cap (S_{1}\cup\ldots\cup S_{k})=\emptyset$. Let $\varphi$ be a
non-negative function such that $\varphi=1$ on $T(r/2)$ zero on $T(r/2)^{c}$
($T^{c}(r/2)$ is the complement of $T(r/2)$ in $\Sigma$) and let $f$ be a
function of support in $T(r)^{c}$. Then we have 
\begin{equation}
\label{NI1} 
\lim \int_{\tilde{S}} f+\varphi\ dS=\sum_{i=1}^{i=k} \int_{\Su_i} f+\varphi\
dS=\sum_{i=1}^{i=k} \int_{\Su_i} f\ dS.  
\end{equation}

On the other hand we have
\begin{equation}
\label{NI2}
\lim \int_{\tilde{S}} f\ dS=\sum_{i=1}^{i=k} \int_{\Su_i} f\ dS,
\end{equation}

\n and
\begin{equation}
\label{NI3}
\lim \int_{\tilde{S}} \varphi\ dS\geq c>0,
\end{equation}

\n for some fixed constant $c>0$ and for every element of the sequence
$\tilde{S}$. This last inequality follows easily from the fact that every
element $\tilde{S}$ must intersect every curve at a distance $d<r/2$ from
$\xi$ (otherwise the intersection number between $\xi$ and
$\tilde{S}$ would be zero, which would imply that the intersection number
between $\xi$ and $S$ would be zero). Inequalities (\ref{NI2}) and (\ref{NI3})
contradict (\ref{NI1}).

\end{proof}

Finally we give the proof of theorem \ref{t:main3}. 

\begin{proof}[Proof of Theorem \ref{t:main3}]
  In \cite{christodoulou88} it has been shown that an isoperimetric stable
  sphere $\Su$ satisfies the following inequality
\begin{equation}
  \label{eq:3}
  12\pi \geq \frac{1}{2} \int_{\Su} R \, dS. 
\end{equation}
Note the extra factor $3$ in comparison with \eqref{e:sta-min}. The left hand
side of (\ref{eq:3}) is bounded in the same way as in the proof of theorem
\ref{PQ2}. 
\end{proof}

We would like to point out that inequalities of the type (\ref{MI2AQ}) are
precursors of further inequalities between mass and charge-squared. Indeed,
using the Riemannian Penrose inequality \cite{Bray01} and Theorem \ref{PQ2} one
can easily prove for instance the following.

\begin{theorem}[Mass, charge and global topology] 
  Let $(\Sigma,(g,K),(E,B))$ be a maximal initial state for the
  Einstein-Maxwell equations, with asymptotically flat ends. Then, for a given
  end $\Sigma_{e}$ we have
\begin{equation}
4 m^{2}\geq \frac{Q_{E}^{2}+Q_{M}^{2}}{|H_{2}|}.
\end{equation}
where $m$ is the mass of $\Sigma_{e}$ and $Q_{E}$ and $Q_{B}$ are its electric
and magnetic charges.
\end{theorem} 

For a different treatment of these type of inequalities see for instance
 \cite{Horowitz84}, \cite{Gibbons82}.

\section*{Acknowledgments}
We would like to thanks W. Simon for illuminating discussions and for pointing
us the relevant reference \cite{Gibbons:1998zr}.

The authors would like to thank the hospitality and support of the Erwin
Schrödinger Institute for Mathematical Physics (ESI), Austria. Part of this
work took place during the program ``Dynamics of General Relativity, Analytical
and Numerical Approaches'', 2011.

Part of this work took also place during the visit of S. D. to the Max Planck
Institute for Gravitational Physics in 2011.  He thanks for the hospitality and
support of this institution.  S. D. is supported by CONICET (Argentina). This
work was supported in part by grant PIP 6354/05 of CONICET (Argentina), grant
Secyt-UNC (Argentina) and the Partner Group grant of the Max Planck Institute
for Gravitational Physics (Germany).  J.L.J. acknowledges the Spanish MICINN
(FIS2008-06078-C03-01) and the Junta de Andaluc\'\i a (FQM2288/219).

%\bibliographystyle{$HOME/biblio/habbrv}
%\bibliography{$HOME/biblio/biblio}

\begin{thebibliography}{10}

\bibitem{Acena:2010ws}
A.~Ace\~na, S.~Dain, and M.~E.~G. Cl\'ement.
\newblock Horizon area--angular momentum inequality for a class of axially
  symmetric black holes.
\newblock {\em Classical and Quantum Gravity}, 28(10):105014, 2011, 1012.2413.

\bibitem{Andersson:2005gq}
L.~Andersson, M.~Mars, and W.~Simon.
\newblock {Local existence of dynamical and trapping horizons}.
\newblock {\em Phys.Rev.Lett.}, 95:111102, 2005, gr-qc/0506013.

\bibitem{andersson08}
L.~Andersson, M.~Mars, and W.~Simon.
\newblock Stability of marginally outer trapped surfaces and existence of
  marginally outer trapped tubes.
\newblock {\em Adv. Theor. Math. Phys.}, 12(4):853--888, 2008.

\bibitem{Ansorg:2010ru}
M.~Ansorg, J.~Hennig, and C.~Cederbaum.
\newblock {Universal properties of distorted Kerr-Newman black holes}.
\newblock {\em Gen.Rel.Grav.}, 43:1205--1210, 2011, 1005.3128.

\bibitem{Ansorg:2007fh}
M.~Ansorg and H.~Pfister.
\newblock {A universal constraint between charge and rotation rate for
  degenerate black holes surrounded by matter}.
\newblock {\em Class. Quant. Grav.}, 25:035009, 2008, 0708.4196.

\bibitem{bonnor98}
W.~B. Bonnor.
\newblock A model of a spheroidal body.
\newblock {\em Classical and Quantum Gravity}, 15(2):351, 1998.

\bibitem{Booth:2006bn}
I.~Booth and S.~Fairhurst.
\newblock {Isolated, slowly evolving, and dynamical trapping horizons: Geometry
  and mechanics from surface deformations}.
\newblock {\em Phys.Rev.}, D75:084019, 2007, gr-qc/0610032.

\bibitem{Booth:2007wu}
I.~Booth and S.~Fairhurst.
\newblock {Extremality conditions for isolated and dynamical horizons}.
\newblock {\em Phys. Rev.}, D77:084005, 2008, 0708.2209.

\bibitem{Bray01}
H.~L. Bray.
\newblock Proof of the riemannian penrose conjecture using the positive mass
  theorem.
\newblock {\em J. Differential Geometry}, 59:177--267, 2001, math.DG/9911173.

\bibitem{Brill63}
D.~R. Brill and R.~W. Lindquist.
\newblock Interaction energy in geometrostatics.
\newblock {\em Phys. Rev.}, 131:471--476, 1963.

\bibitem{christodoulou88}
D.~Christodoulou and S.-T. Yau.
\newblock Some remarks on the quasi-local mass.
\newblock In {\em Mathematics and general relativity ({S}anta {C}ruz, {CA},
  1986)}, volume~71 of {\em Contemp. Math.}, pages 9--14. Amer. Math. Soc.,
  Providence, RI, 1988.

\bibitem{Chrusciel02a}
P.~T. Chru{\'s}ciel and R.~Mazzeo.
\newblock On `many-black-hole' vacuum spacetimes.
\newblock {\em Class. Quantum. Grav.}, 20(4):729--754, 2003, gr-qc/0210103.

\bibitem{dain10d}
S.~Dain.
\newblock Extreme throat initial data set and horizon area-angular momentum
  inequality for axisymmetric black holes.
\newblock {\em Phys. Rev. D}, 82(10):104010, Nov 2010, 1008.0019.

\bibitem{Dain:2011pi}
S.~Dain and M.~Reiris.
\newblock Area---angular-momentum inequality for axisymmetric black holes.
\newblock {\em Phys. Rev. Lett.}, 107(5):051101, Jul 2011, 1102.5215.

\bibitem{Dain:2010pj}
S.~Dain, G.~Weinstein, and S.~Yamada.
\newblock {Counterexample to a Penrose inequality conjectured by Gibbons}.
\newblock {\em Class.Quant.Grav.}, 28:085015, 2011, 1012.4190.

\bibitem{Gibbons:1998zr}
G.~Gibbons.
\newblock {Some comments on gravitational entropy and the inverse mean
  curvature flow}.
\newblock {\em Class.Quant.Grav.}, 16:1677--1687, 1999, hep-th/9809167.

\bibitem{Gibbons82}
G.~W. Gibbons and C.~M. Hull.
\newblock A {B}ogomolny bound for general relativity and solitons in {$N=2$}\
  supergravity.
\newblock {\em Phys. Lett. B}, 109(3):190--194, 1982.

\bibitem{Guillemin-Pollack}
V.~Guillemin and A.~Pollack.
\newblock {\em Differential topology}.
\newblock Prentice-Hall Inc., Englewood Cliffs, N.J., 1974.

\bibitem{Hatcher02}
A.~Hatcher.
\newblock {\em Algebraic topology}.
\newblock Cambridge University Press, Cambridge, 2002.

\bibitem{hayward94}
S.~A. Hayward.
\newblock General laws of black-hole dynamics.
\newblock {\em Phys. Rev. D}, 49:6467--6474, Jun 1994.

\bibitem{hennig08}
J.~Hennig, M.~Ansorg, and C.~Cederbaum.
\newblock A universal inequality between the angular momentum and horizon area
  for axisymmetric and stationary black holes with surrounding matter.
\newblock {\em Class. Quantum. Grav.}, 25(16):162002, 2008.

\bibitem{Hennig:2008zy}
J.~Hennig, C.~Cederbaum, and M.~Ansorg.
\newblock {A universal inequality for axisymmetric and stationary black holes
  with surrounding matter in the Einstein-Maxwell theory}.
\newblock {\em Commun. Math. Phys.}, 293:449--467, 2010, 0812.2811.

\bibitem{Horowitz84}
G.~T. Horowitz.
\newblock The positive energy theorem and its extensions.
\newblock In F.~J. Flaherty, editor, {\em Asymptotic behavior of mass and
  spacetime geometry (Corvallis, Ore., 1983)}, volume 202 of {\em Lecture Notes
  in Phys.}, pages 1--21. Springer, Berlin, 1984.

\bibitem{Jang79}
P.~S. Jang.
\newblock Note on cosmic censorship.
\newblock {\em Phys. Rev. D}, 20(4):834--837, 1979.

\bibitem{Jaramillo:2011pg}
J.~L. Jaramillo, M.~Reiris, and S.~Dain.
\newblock {Black hole Area-Angular momentum inequality in non-vacuum
  spacetimes}, 2011, 1106.3743.

\bibitem{Meeks82}
W.~Meeks, III, L.~Simon, and S.~T. Yau.
\newblock Embedded minimal surfaces, exotic spheres, and manifolds with
  positive {R}icci curvature.
\newblock {\em Ann. of Math. (2)}, 116(3):621--659, 1982.

\bibitem{Racz:2008tf}
I.~Racz.
\newblock {A Simple proof of the recent generalisations of Hawking's black hole
  topology theorem}.
\newblock {\em Class.Quant.Grav.}, 25:162001, 2008, 0806.4373.

\end{thebibliography}

\end{document}